\documentclass[doublecol]{epl2} 

\usepackage{graphicx,color}
\usepackage{amssymb}
\usepackage{amsmath,amsthm,algorithmic,algorithm}
\usepackage{epstopdf}
\newtheorem{theorem}{Theorem}
\newtheorem{lemma}{Lemma}

\newtheorem{definition}{Definition}

\title{Active influence in dynamical models of \\structural balance in social networks}
\shorttitle{Structural balance in social networks} 

\author{Tyler H. Summers\inst{1} \and Iman Shames\inst{2}}
\shortauthor{T.H. Summers and I. Shames}

\institute{                    
  \inst{1} Automatic Control Laboratory, ETH Zurich\\
  \inst{2} Department of Electrical and Electronic Engineering, University of Melbourne
}


\abstract{
We consider a nonlinear dynamical system on a signed graph, which can be interpreted as a mathematical model of social networks in which the links can have both positive and negative connotations. In accordance with a concept from social psychology called \emph{structural balance}, the negative links play a key role in both the structure and dynamics of the network. Recent research has shown that in a nonlinear dynamical system modeling the time evolution of ``friendliness levels'' in the network, two opposing factions emerge from almost any initial condition. Here we study active external influence in this dynamical model and show that any agent in the network can achieve any desired structurally balanced state from any initial condition by perturbing its own local friendliness levels. Based on this result, we also introduce a new network centrality measure for signed networks. The results are illustrated in an international relations network using United Nations voting record data from 1946 to 2008 to estimate friendliness levels amongst various countries.}

\begin{document}

\maketitle

\section{Introduction}
There has been a broad effort in many fields of science and engineering to understand (and eventually optimize and control) large networks, consisting of many interacting subsystems. One of the key goals of this effort is to describe and quantify how graph properties of the interconnection structure interrelate with properties of dynamical processes on the network. However, there has been relatively little work on structure and dynamics on signed graphs, in which interconnections in the networks can have either a positive or a negative association. Signed graphs arise in models for a variety of systems, including social networks, data classification and clustering, genetic regulatory networks, and non-ferromagnetic Ising models. In this paper, we study a nonlinear dynamical system on a signed graph, which arises from a mathematical model of social networks in which the links can have both positive and negative connotations.

In mathematical models for social network analysis, links in the network often have positive connotations, such as friendship, collaboration, information sharing, etc. However, negative interactions in social networks, such as antagonism, distrust, or disagreement, also play a key role in both structure and dynamics of social networks and are receiving increased attention in the literature \cite{easley2010networks,leskovec2010signed,leskovec2010predicting}. They are particularly interesting in light of new online networks that provide real-world dynamical data for social networks with both positive and negative relationships. For example, users on the product review website Epinions can display both trust and distrust of other users; on the technology news website Slashdot, users can designate other users as either friend or foe; and on Wikipedia, users can vote for or against another person becoming an administrator. The structure of such networks has been studied in \cite{leskovec2010predicting,leskovec2010signed}. 

The concept of \emph{structural balance} is an old idea in sociology, tracing back to social psychology research in the 1940s by Heider \cite{heider1946attitudes}. The theory begins with notions of tension and balance in three-agent networks. Imagine that a person has two good friends who hate each other. There is a tension in this situation that is resolved when either the person takes one side and ends the friendship with the other or when the feuding friends reconcile their differences. Similarly, there is a tension amongst three people unfriendly with one another that is resolved when two of them to form an alliance against the other. The theory was generalized to networks of $n$ agents in the 1950s by Cartwright and Harary \cite{cartwright1956structural}, who model the network as a complete signed graph in which the $n$ vertices represent agents and a complete signed edge set represents relationships amongst all agents in the network, with positive and negative signs associated with friendly and hostile relationships, respectively. They showed that $n$-agent structurally balanced networks are precisely those that can be partitioned into two factions, such that within each faction all relationships are friendly and between factions all relationships are hostile. The theory has found various applications, e.g., in models of international relations \cite{axelrod1991landscape}, but has remained static, focusing mostly on network structure.

Dynamic models for structural balance are quite recent and provide a new and interesting perspective \cite{antal2005dynamics,kulakowski2005heider,marvel2011continuous,krawczyk2010application}. If a network is in an initial state that is not structurally balanced, how might the state of the network evolve toward a structurally balanced state, and what will the eventual balanced state be? Discrete dynamical models, in which a relationship is either positive or negative, have been proposed in \cite{antal2005dynamics} and \cite{marvel2009energy}. In these models, the system evolves by flipping the sign on certain edges to increase the number structurally balanced triangles in the network. However, these models suffer from the existence of so-called ``jammed states'', in which the system becomes stuck in a structurally unbalanced local minimum. More recently, Ku{\l}akowski et al. \cite{kulakowski2005heider} and Marvel et al. \cite{marvel2011continuous} have proposed and analyzed a continuous dynamical model, in which a real-valued ``friendliness level'' is associated with each relationship (positive values indicate friendliness and negative values indicate hostility). In this model, Marvel et al. \cite{marvel2011continuous} show that for generic initial conditions, the system converges to a structurally balanced state in finite time. Further, a closed-form expression for the balanced state is derived in terms of the initial state.

In this paper, we study active external influence in a dynamic model of structural balance. In particular, we suppose that a single agent can influence the state of the network by perturbing its own friendliness levels; in an international relations context, this might be the result of certain foreign policy actions. We show that it is possible for any single agent to achieve any desired structurally balanced state given any initial state. We also present a method to compute and optimize the influence that is required to achieve the desired state. The magnitude of the required input defines a new dynamic network centrality measure for signed graphs in the context of structural balance, assigning a relative ``influence'' value to each agent in the network in terms of its ability to achieve a desired structurally balanced state with a low-magnitude perturbation. An agent who can achieve a desired state with a perturbation of smaller magnitude is more influential than agents who require larger perturbations. We illustrate the results to structural balance in an international relations network, using data from United Nations General Assembly voting records dating from 1946 to 2008 to estimate friendliness levels amongst various countries. Our results provide an interesting lens through which one can view the data.

The paper is organized as follows. Section 2 reviews the model and basic results from \cite{marvel2011continuous} and \cite{kulakowski2005heider}. Section 3 presents our main results. In Section 4, we illustrate our results to an international relations case study. Finally, concluding remarks and a future research outlook are provided in Section 5.

\section{A Continuous Dynamical Model for Structural Balance}
In this section, we review the continuous dynamic model for structural balance proposed in \cite{kulakowski2005heider} and \cite{marvel2011continuous} and a basic result from \cite{marvel2011continuous}. Let $x_{ij}$ denote the ``friendliness level'' between agents $i$ and $j$. We allow $x_{ii}$ to be non-zero, which can be interpreted as a ``self-confidence'' or ``willingness to compromise'' of node $i$ (large positive value means high self-confidence or low willingness to compromise, large negative value means low self-confidence or high willingness to compromise). We assume that relationships are symmetric, i.e., $x_{ij} = x_{ji}$, and collect the friendliness levels in the network into a symmetric matrix $X\in \mathbf{R}^{n\times n}$. For any given $X\in \mathbf{R}^{n\times n}$, we associate a complete\footnote{In this paper we assume that each agent has (or develops) an opinion about every other agent, which means that $X$ is assumed to have no non-zero entries. The theory of structural balance can be extended to non-complete graphs; see e.g. \cite{easley2010networks}. We will consider this in future work.} signed graph in which the edge signs correspond to the signs of each element $x_{ij}$ of $X$. A signed graph is called \emph{structurally balanced} if all triangles are balanced. We say that the network is in a structurally balanced state if the corresponding complete signed graph is structurally balanced. 

The dynamic model proposed in \cite{kulakowski2005heider} and analyzed in \cite{marvel2011continuous} is given by the matrix differential equation
\begin{equation} \label{dyn}
\dot{X}(t) = X^2(t), \quad X(0) = X_0
\end{equation}
or equivalently elementwise by
\begin{equation} \label{dynelement}
\dot{x}_{ij}(t) = \sum_k x_{ik}(t) x_{kj}(t).
\end{equation}
Each term in (\ref{dynelement}) moves the associated triangle $ijk$ toward a balanced configuration, and for a given pair $ij$ the summation aggregates the effects across all relationships of $i$ and $j$ in the network. This can be interpreted as a gossip process, in which the friendliness level in a relationship is changed based on opinions of mutual friends. 

One of the main results from \cite{marvel2011continuous}, which is relevant for our study, is the following:
\begin{theorem}[\cite{marvel2011continuous}] \label{dynbalance}
Suppose $X_0$ is a random initial matrix, with entries sampled independently from an absolutely continuous distribution with bounded support. Let $\lambda_1 \geq \cdots \geq \lambda_n$ and $w_1,... ,w_n$ be respectively the eigenvalues and associated eigenvectors of $X_0$. Then with probability converging to 1 in the number of agents $n$ we have
\begin{itemize} \itemsep .01cm
\item $\lambda_1 > 0$,
\item $\lambda_1 \neq \lambda_2$,
\item all components $w_1$ are nonzero.
\end{itemize}
As a consequence, the network converges with finite escape time $t^* = 1/\lambda_1$ to a structurally balanced state with the same sign pattern as the rank one matrix $w_1 w_1^T$; thus, the two factions are determined by the sign pattern of $w_1$: $\mathcal{E} = \{k: w_{1k} > 0\}$ and $\mathcal{R} = \{ k: w_{1k} < 0 \}$. 
\end{theorem}
Theorem 1 states that the network converges to a structurally balanced state for almost all initial conditions. The differential equation (\ref{dyn}) in fact has an analytical solution, which can be obtained by diagonalizing the initial state matrix and analytically solving a scalar differential equation of the same form. The solution of (\ref{dyn}) is given by 
\begin{equation}
X(t) = X_0(I - X_0 t)^{-1},
\end{equation}
which is valid for $t < 1/\lambda_1$ if $\lambda_1 > 0$. This indicates that system ``blows up'' with a finite escape time determined by the largest eigenvalue $\lambda_1$ of $X_0$; the $x_{ij}$'s associated with each faction converge to either plus or minus infinity.  However, the system matrix normalized by the Frobenius norm, viz. $X/\Vert X \Vert_F$, effectively collapses to a rank one matrix defined by the eigenvector associated with the largest eigenvalue of $X_0$. In the context of the applications that we discuss later, the implication of Theorem \ref{dynbalance} is that the sign pattern of the largest eigenvector gives a \emph{prediction} of the eventual factions that emerge based on the current state of the network.
An example trajectory with $n=50$ and the entries of $X_0$ drawn from a uniform distribution on $[-1,1]$ is shown in Figure 1.
\begin{figure} \label{extraj}
\begin{center} 
\resizebox{\linewidth}{!}{\includegraphics{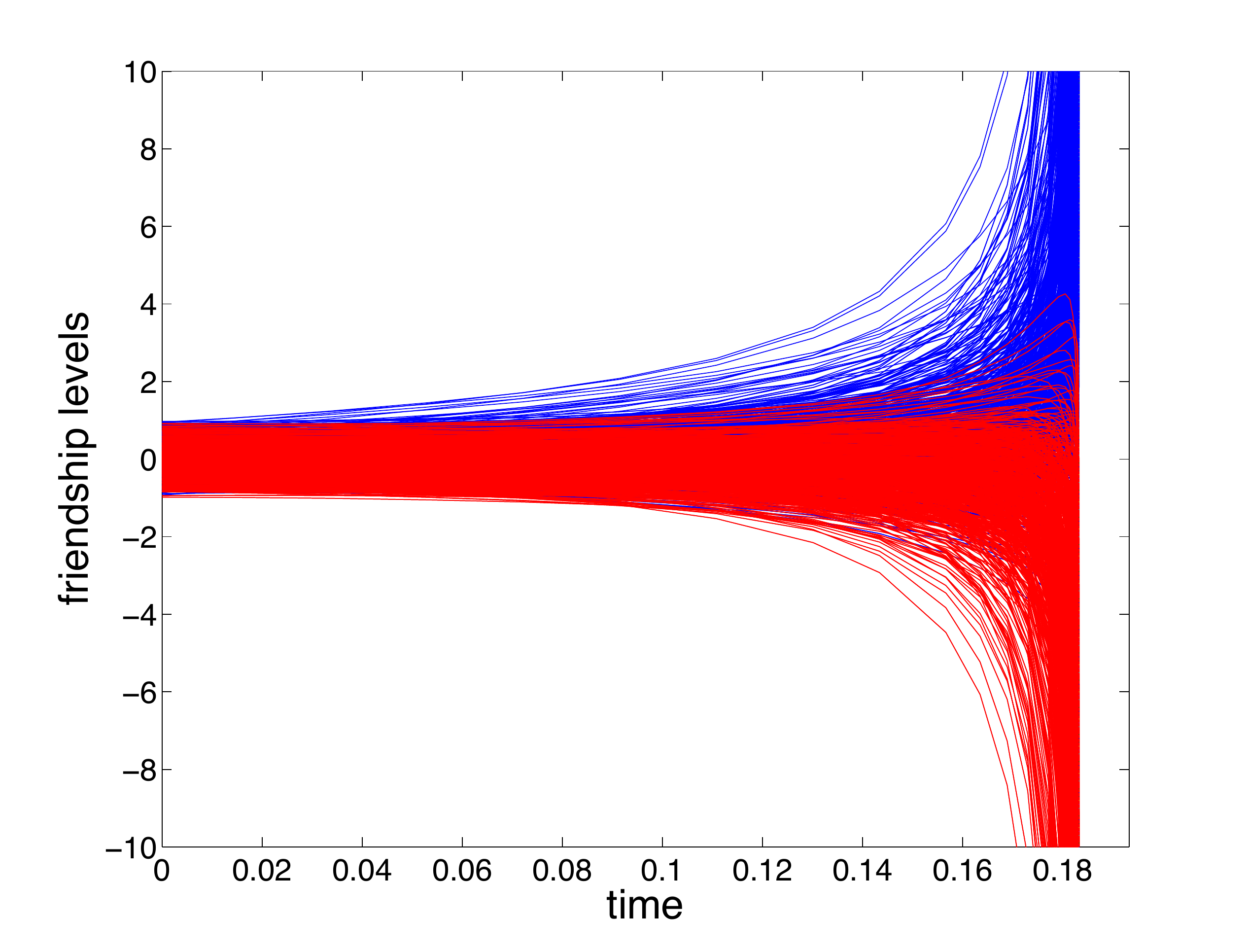}}
\caption{Network friendliness trajectories for a random initial friendliness matrix $X_0$. The friendliness levels ``blow up'' to plus or minus infinity, defining the two opposing factions that emerge.}
\end{center}
\end{figure}


%

\section{Active Influence in Dynamic Models of Structural Balance}
In the last section, we saw that for the dynamic model (\ref{dyn}), the system converges to a structurally balanced state for generic initial conditions, and that this state can be determined from an eigenvector of the initial state matrix.  We now study active external influence influence in this model. We consider a situation in which a single agent can instantaneously perturb its own friendliness levels. We will show that any single agent can produce such a perturbation to achieve any structurally balanced state from any initial condition.

\subsection{Main Result: Single Agent Influence}
Our single agent influence model means that, given an initial state $X_0$, agent $i$ can choose an symmetric perturbation $\Delta X_0$, which has entries equal to zero except for entries in the $i$th row and column, that is added to the initial state to produce a new state $X_0+\Delta X_0$ from which the dynamics flow according to \eqref{dyn}. We will prove the following main result.

 \begin{theorem}\label{main}
Let $X_0 \in \mathbf{R}^{n\times n}$ be any symmetric initial state matrix and $v^*$ be a vector with a desired sign pattern whose entries are either $1$ or $-1$. Then there exists a symmetric perturbation matrix $ \Delta X_0$, computable from $X_0$ and $v^*$, with entries equal to zero except for entries in the $i$th row and column such that the eigenvector associated with the largest eigenvalue of $X_0 + \Delta X_0$ has the same sign pattern as $v^*$.
\end{theorem}
Before proving Theorem \ref{main}, we present the following lemma on the eigenvalues of the structured perturbation $\Delta X_0$.
\begin{lemma}\label{lem:DX}
A symmetric $  \Delta X_0$ with all zero entries except the entries of the $i$th row and column has exactly two nonzero eigenvalues with opposite signs.
\end{lemma}
\begin{proof}
Without loss of generality, assume $i=1$. Then $\Delta X_0$ has the form
$$
\Delta X_0=\left [\begin{array}{cc} \delta x_{1} & \overline{\delta x}^\top\\ \overline{\delta x} & 0 \end{array}\right ]
$$
where $\delta x=[\delta x_{1},\delta x_2,\dots,\delta x_n]^\top$ is the first column of $\Delta X_0$, and $\overline{\delta x}=[\delta x_2,\dots,\delta x_n]^\top$. The eigenvalues of $\Delta X_0$ are the solutions of 
\begin{equation}\label{eq:det-DX}
\begin{aligned}
\det( \mu I-\Delta X_0)&=0\\
\det \left [\begin{array}{cc} \mu-\delta x_{1} & -\overline{\delta x}^\top\\ -\overline{\delta x} & -\mu I \end{array}\right ] &=0.
\end{aligned}
\end{equation}
Using the formula for block determinants $\det A=(a_{11} -a_{21}^T A_{22}^{-1}a_{21}) \det A_{22} $ where $A=\left [\begin{array}{cc} a_{11} & a_{21}^T \\ a_{21} & A_{22} \end{array}\right ]$, \eqref{eq:det-DX} becomes
\begin{equation}
\mu^{n-2}(\mu(\mu-\delta x_1)-\overline{\delta x}^\top \overline{\delta x})=0.
\end{equation}
Thus, the eigenvalues are $\mu=\frac{\delta x_1\pm \sqrt{\delta x_1^2+4\overline{\delta x}^\top \overline{\delta x}}}{2}$, or $\mu=0$, which proves the result.
\end{proof}

\textit{Proof of Theorem \ref{main}.}
Without loss of generality let $i=1$. We seek a perturbation $\Delta X_0$ such that
\begin{equation}
(X_0+\Delta X_0) v^* = \lambda^* v^*.
\end{equation}
For a given $X_0$, $v^*$, and $\lambda^*$, this is a system of linear equations in $\Delta X_0$, which has the solution
\begin{equation} \label{pert}
\delta x_0^* = V^{-1}(\lambda^* I - X_0)v^*,
\end{equation}
where
\begin{equation} \label{vinv}
V^{-1}=\frac{1}{v^*_1}\left[\begin{array}{cccc}1 & -\frac{v^*_2}{v^*_1} &\dots  & -\frac{v^*_n}{v^*_1} \\ & 1 &  &  \\ &  & \ddots &  \\ &  &  & 1\end{array}\right].
\end{equation}
It remains to be shown that $\lambda^*$  and $v^*$ correspond to the largest eigenvalue and eigenvector pair of $X_0+\Delta X_0$. This can be done by taking $\lambda^* \geq \lambda_1(X_0)$ and using a basic spectral properties (called Weyl inequalities \cite{weyl1912asymptotische}) of the sum of two symmetric matrices. In particular, we have
 \begin{equation}
\lambda_i(X+\Delta X_0) \leq \lambda_1(X_0)+\lambda_i(\Delta X_0), \quad i=1,\dots,n.
 \end{equation}
Taking $i=2$ and noting that $\lambda_2(\Delta X_0) = 0$ from Lemma 1, we have $\lambda_2(X+\Delta X_0) \leq \lambda_1(X_0)$ and thus $\lambda^* = \lambda_1(X+\Delta X_0) \geq \lambda_1(X_0)$, which proves the result. $\qed$

\subsection{Optimizing the Influence}
Equations \eqref{pert} and \eqref{vinv} in the proof of Theorem \ref{main} give an explicit expression for the perturbation required to yield an eigenvector with a specified sign pattern $v^*$ in terms of the initial state matrix $\Delta X_0$ and the largest eigenvalue of the perturbed matrix $\lambda^*$. In real networks, there is a cost associated perturbing the state of the network, and so we are interested in finding perturbations with small magnitude. In particular, we would like to solve the following optimization problem:
\begin{equation}
\begin{aligned}
& \underset{\lambda^*, v^*}{\text{minimize}} && \Vert \delta x_0 \Vert =  \Vert V^{-1}(\lambda^* I - X_0)v^* \Vert  \\
& \text{subject to} && \lambda^* \geq \lambda_1(X_0),
\end{aligned}
\end{equation}
where the entries of $v^*$ can be freely chosen up to the specified sign pattern.

Unfortunately, this problem is not convex; however, we can obtain an approximate solution as follows. First, note that for fixed $v^*$, the objective is minimized when $\lambda^* = \lambda_1(X_0)$, since the objective is monotone increasing in $\lambda^*$. Next, we can obtain an upper bound for the solution and choose $v^*$ to minimize the upper bound. Let $\alpha_{i} = \frac{v_i^*}{v_1^*}$, $\alpha = [\alpha_2,\dots,\alpha_n]^T$, $L = \lambda^* I-X_0$, $L_1$ be the first column of $L$, and $\overline{L}$ be the principal  submatrix of $L$ obtained by removing the first row and column of $L$. Without loss of generality let $v^\star_1$ be positive. An upper bound for the objective is
\begin{equation} \label{eq:up_bound}
\begin{aligned}
\Vert \delta x_0 \Vert &=  \Vert V^{-1}(\lambda^* I - X_0)v^* \Vert = \left\| L_1 +  \left [\begin{array}{c}- \alpha^T \overline{L}\alpha \\ \overline{L} \alpha \end{array}\right] \right\| \\
				&\leq \Vert L_1 \Vert + \left\| \left [\begin{array}{c}- \alpha^T \overline{L}\alpha \\ \overline{L} \alpha \end{array}\right]  \right\|
\end{aligned}
\end{equation}
The upper bound in \eqref{eq:up_bound} is minimized when $\alpha=0$, and this in turn is obtained asymptotically as $v_1^*$ becomes much larger than $v_i^*$ for $i=2,\dots,n$. Thus, taking $\lambda^* = \lambda_1(X_0)$, $v_1^* = 1$ (without loss of generality), and $v_i^*$ small but with the prescribed sign yields a feasible perturbation with small norm. 



\subsection{A Network Centrality Measure for Signed Graphs}
Network centrality measures are real-valued functions that assign a relative ``importance'' to each vertex within the graph. Examples include degree, betweenness, closeness, and eigenvector centrality, among others. The meaning of centrality or importance and the relevance of various metrics depends highly on the modeling context. For example, PageRank, a variant of eigenvector centrality, turns out to be a much better indicator of importance than vertex degree in the context of networks of web pages, one of the key factors leading to Google's domination of web search.

In the context of signed graphs and structural balance, the magnitude of the input required to achieve a desired structurally balanced state, as described in the previous section defines a class of network centrality measures. In particular, this input magnitude gives an importance or influence value for each agent that measures how easily an agent can perturb the network into the desired structurally balanced state. For a given desired state, this magnitude can be used to rank the influence of all agents in the network by computing the value for each agent and sorting the result; the smallest magnitude corresponds to the most influential agent. The required input magnitude and the ranking depends of course on the desired state. We thus define the Structural Balance Influence Index for signed graphs, parameterized by the desired structurally balanced state, as follows.
\begin{definition}[Structural Balance Influence Index]
Given a signed graph with $n$ vertices and an associated symmetric friendliness matrix $X \in \mathbf{R}^{n \times n}$, let $X_i$ be a permutation of $X$ with row and column $i$ moved to the first row and column. Let $v^* \in \{-1,1\}^n$ be a sign pattern vector that defines a desired structurally balanced state. The Structural Balance Influence Index of agent $i$ given $v^*$ is the norm of the input from agent $i$ required to achieve $v^*$:
\begin{equation}
SBII_{v^*} (i) =  \Vert V^{-1}(\lambda_1(X_i)I - X_i)\hat{v}^* \Vert,
\end{equation}
where $V^{-1}$ is given by \eqref{vinv} and $\hat{v}^{*T} = v^{*T} \bullet [1, \epsilon, ..., \epsilon]^T$, with $\epsilon$ small and positive and $\bullet$ denoting the Hadamard (element-wise) product.
\end{definition}

To the best of our knowledge, this is the first such centrality measure defined in this context. It would be interesting to study the distribution of this measure in random and real-world networks.

\section{Illustrative Example}
In this section, we interpret our results in an international relations network. We use UN General Assembly voting records and gross domestic product (GDP) history of 45 countries from 1946 to 2008 to estimate friendliness and economic importance levels amongst the countries. The UN voting records are obtained from \cite{voeten2009united}. As a metric of friendliness levels, we use an affinity index that has appeared in the international relations literature, which is computed for each year based on how often nations vote together in the general assembly; see  \cite{gartzke2006affinity} for details. Finally, we multiply the affinity index associated with each pair of countries with the corresponding gross domestic products (GDPs) of the two countries in order to capture the relative economic importance of each relationship.

We emphasize that the dynamic model \eqref{dyn} is not intended to be a detailed model for international relations (only UN voting record similarities and GDPs are accounted for); rather the example is used only to illustrate and interpret our results. Given a current estimate of the friendliness levels in the network, Theorem 1 can be used to predict eventual factions that will emerge by calculating the eigenvector associated with the largest eigenvalue of the friendliness matrix.  Figure 2 shows the factions predicted by the model over time. The faction containing the United States is blue, and the opposing faction is red. \

An interesting question in this context is: how could various countries in the network adjust their friendliness bilaterally so that the model predicts global harmony (i.e., all countries converge to a single faction)? Also, which countries require the smallest perturbation to their bilateral relationships to achieve this? These question can be answered in our framework by fixing the desired structurally balanced state $v^*$ to have entries with the same sign and computing the Structural Balance Influence Index for each country.  Figure 3 shows how the Structural Balance Influence Index for each country evolves over time. We see that the United States is the most influential agent in the network according to this model, due to the relative size of its GDP and to its position in the network. From 1946 to 1990, the USSR/Russia required the largest change in its bilateral relationships to achieve global harmony and could be viewed as the economic leader of the faction opposing the United States. In 1990, this role was taken over by China. It is also interesting to note that there has been a global increase in the ability of any country to achieve global harmony over the last eight years of the data, reflecting increasing disagreement in United Nations voting.   

\begin{figure*}
\begin{center} 
\resizebox{1\linewidth}{!}{\includegraphics{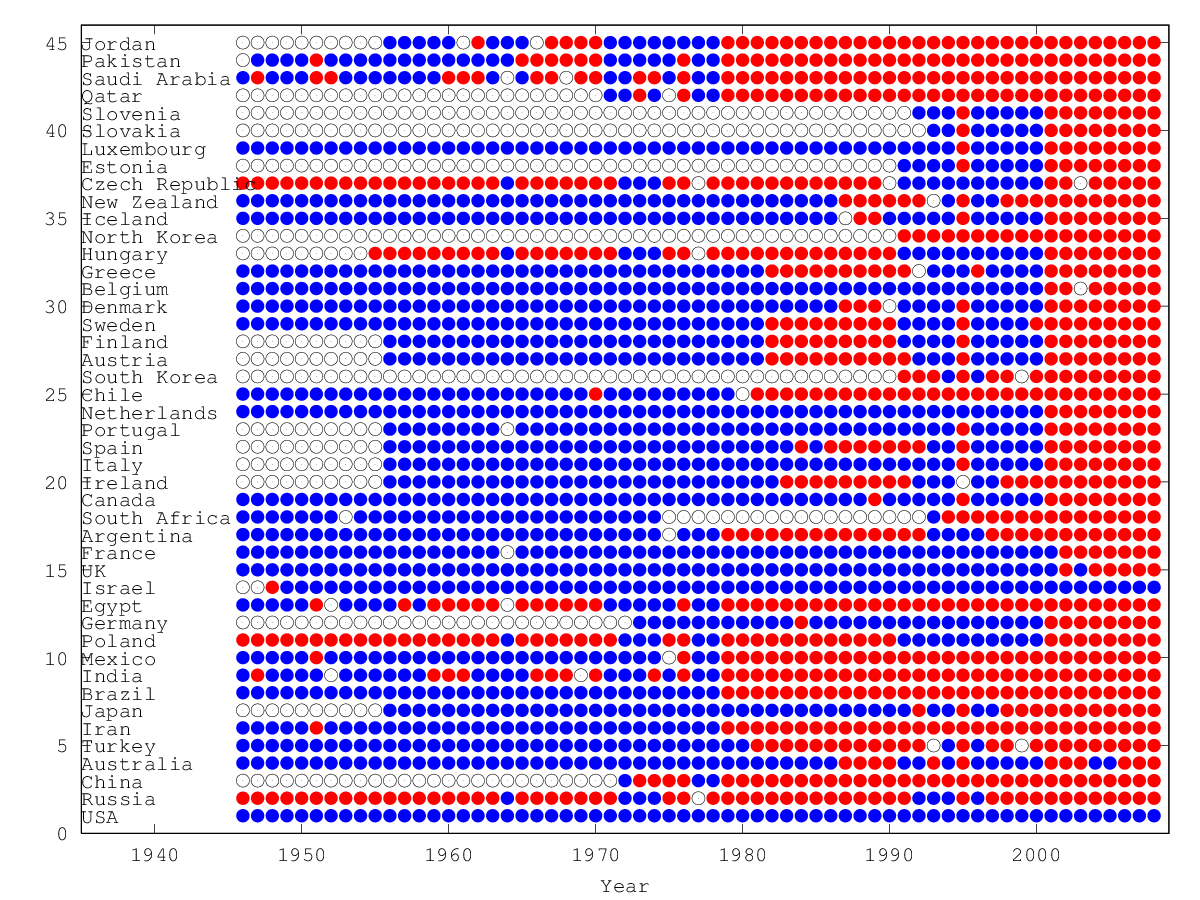}}
\caption{Predicted factions by year.}
\end{center}
\end{figure*}



\begin{figure}
\begin{center} 
\resizebox{1\linewidth}{!}{\includegraphics{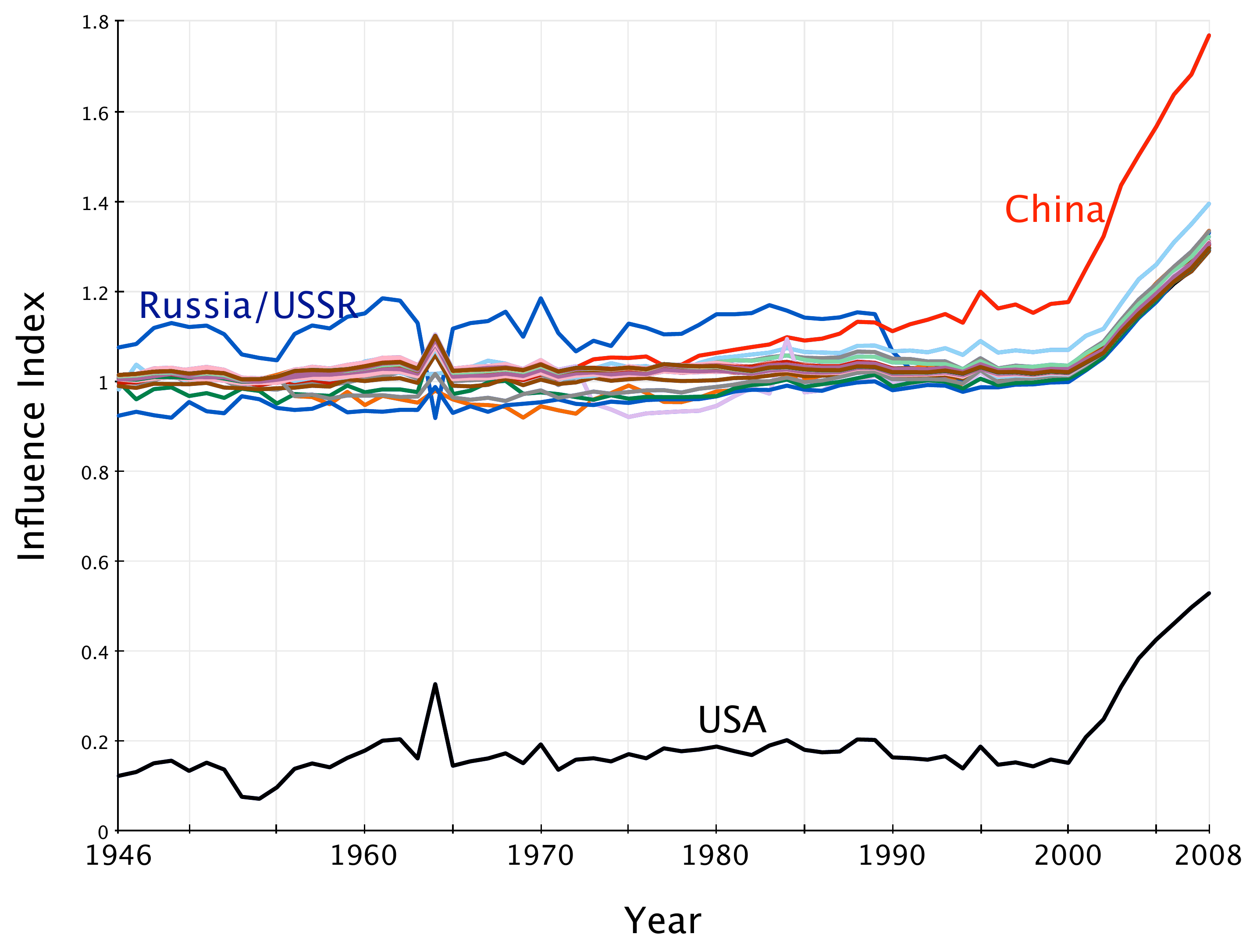}}
\caption{Magnitude of influence required to achieve global harmony (all agents converge to a single faction) over time. The combination of large GDP and position in the network makes the United States the most influential agent using this model. We also see a global increase in the ability to achieve global harmony over the last 8 years.}
\end{center}
\end{figure}

\section{Conclusions}
We have considered active influence in a dynamic model of structural balance in social networks with both positive and negative links. We showed that any single agent can locally perturb the state of the network to achieve any desired structurally balanced state from any initial state. We also showed how to optimize the required influence; the magnitude of the required influence defines a network centrality measure in the context of signed graphs and structural balance. The results were illustrated and interpreted in an international relations network based on United Nations voting data from 1946-2008. The model and our results give an interesting lens through which one can view the voting data.


There are a variety of potential extensions for future research. One could consider extending the results to graphs that are not complete, or to directed graphs in which friendship levels are not necessarily mutual. One could also consider models that allow the emergence of more than two factions; these types of models are studied in \cite{krawczyk2010application}. 
Another interesting direction would be to consider multiple agent having active influence, resulting in a game theoretic formulation. In an international relations context, analysis of a game between two super powers would be particularly interesting. Finally, one could explore interpretations of the results here in signed graphs arising in other contexts, such as data clustering or genetic regulatory networks.
\bibliographystyle{plain}  
\bibliography{refs}  

\end{document}